\documentclass[a4paper,10pt]{article}
\usepackage{multirow}
\usepackage{amssymb}
\usepackage{amsthm}
\usepackage[utf8]{inputenc}
\usepackage{array}
\usepackage[T1]{fontenc}
\usepackage{lmodern}
\usepackage{tikz}
\usepackage{vmargin}
\usetikzlibrary{arrows}
\usetikzlibrary{shapes}
\setmarginsrb{3.5cm}{2cm}{3.5cm}{2cm}{1cm}{1cm}{1cm}{1cm}

\title{Partitioning a triangle-free planar graph into a forest and a
  forest of bounded degree}

\usepackage{hyperref}
\usepackage[affil-it]{authblk}

\author[a]{François Dross}
\author[a]{Mickael Montassier}
\author[a,b]{Alexandre Pinlou}

\affil[a]{{\small Université de Montpellier, CNRS, LIRMM}} 
\affil[b]{{\small Université Paul-Valéry Montpellier 3, Département MIAp\medskip}} \affil[ ]{{\small 161 rue Ada, 34095 Montpellier
    Cedex 5, France}} \affil[
]{\href{mailto:francois.dross@lirmm.fr,mickael.montassier@lirmm.fr,alexandre.pinlou@lirmm.fr}{\small{\{francois.dross,mickael.montassier,alexandre.pinlou\}@lirmm.fr}}}

\begin{document}

\maketitle
\newtheorem{theo}{Theorem}
\newtheorem*{theo*}{Theorem}
\newtheorem{cor}[theo]{Corollary}
\newtheorem{lemm}[theo]{Lemma}
\newtheorem{prop}[theo]{Property}
\newtheorem{obs}[theo]{Observation}
\newtheorem{conj}[theo]{Conjecture}
\newtheorem{claim}[theo]{Claim}
\newtheorem{config}[theo]{Configuration}
\newtheorem{quest}[theo]{Question}
\begin{abstract}

  An $({\cal F},{\cal F}_d)$-partition of a
  graph is a vertex-partition into two sets $F$ and $F_d$ such
  that the graph induced by $F$ is a forest and the one induced by
  $F_d$ is a forest with maximum degree at most $d$. We prove that
  every triangle-free planar graph admits an
  $({\cal F},{\cal F}_5)$-partition. Moreover we show that if for some
  integer $d$ there exists a triangle-free planar graph that does not
  admit an $({\cal F},{\cal F}_d)$-partition, then it is an
  NP-complete problem to decide whether a triangle-free planar graph admits
  such a partition.
\end{abstract} 

\section{Introduction}
We only consider finite simple graphs, with neither
loops nor multi-edges. Planar graphs we consider are supposed to be
embedded in the plane.
Consider $i$ classes of graphs ${\cal G}_1,\ldots, {\cal G}_i$. A
$({\cal G}_1,\ldots, {\cal G}_i)$-partition of a graph $G$ is a
vertex-partition into $i$ sets $V_1,\ldots,V_i$ such that, for all
$1\le j \le i$, the graph $G[V_j]$ induced by $V_j$ belongs to ${\cal
  G}_j$. In the following we will consider the following classes of
graphs:
\begin{itemize}
\item ${\cal F}$ the class of forests,
\item ${\cal F}_d$ the class of forests with maximum degree at most
  $d$,
\item ${\cal D}_d$ the class of $d$-degenerate graphs (recall that a
  \emph{$d$-degenerate graph} is a graph such that all subgraphs have a
  vertex of degree at most $d$),
\item ${\Delta}_d$ the class of graphs with maximum degree at most $d$,
\item ${\cal I}$ the class of empty graphs (i.e. graphs with no edges).
\end{itemize}
For example, an $({\cal I},{\cal F},{\cal D}_2)$-partition of $G$ is a
vertex-partition into three sets $V_1,V_2,V_3$ such that $G[V_1]$ is
an empty graph, $G[V_2]$ is a forest, and $G[V_3]$ is a 2-degenerate
graph.
%

The Four Colour Theorem \cite{appel1,appel2} states that every planar
graph $G$ admits a proper $4$-colouring, that is that $G$ can be
partitioned into four empty graphs, i.e. $G$ has an
$({\cal I}, {\cal I}, {\cal I}, {\cal
  I})$-partition.
Borodin~\cite{Borodin} proved that every planar graph admits an
acyclic colouring with at most five colours (an acyclic colouring is a
proper colouring in which every two colour classes induce a
forest). This implies that every planar graph admits an
$({\cal I}, {\cal F}, {\cal F})$-partition. Poh~\cite{Poh} proved that
every planar graph admits an
$({\cal F}_2, {\cal F}_2, {\cal F}_2)$-partition. Thomassen proved
that every planar graph admits an
$({\cal F}, {\cal D}_2)$-partition~\cite{thomassen1995decomposing},
and an
$({\cal I}, {\cal
  D}_3)$-partition~\cite{thomassen2001decomposing}.
However, there are planar graphs that do not admit any
$({\cal F}, {\cal F})$-partition~\cite{chartrand1969point}. Borodin
and Glebov~\cite{borodin2001partition} proved that every planar graph
of girth at least $5$ (that is every planar graph with no triangles
nor cycles of length $4$) admits an $({\cal I}, {\cal F})$-partition.

We focus on triangle-free planar graphs. Raspaud and Wang~\cite{raspaud2008vertex} proved that every planar graph with no triangles at distance at most $2$ (and thus in particular every triangle-free planar graph) admits an $({\cal F}, {\cal F})$-partition. However, it is not known whether every triangle-free planar graph admits an $({\cal I}, {\cal F})$-partition. We pose the following questions: 

\begin{quest} \label{q1}
  Does every triangle-free planar graph admit an $({\cal I}, {\cal F})$-partition?
\end{quest}

\begin{quest} \label{q2}
  More generally, what is the lowest $d$ such that every triangle-free planar graph admits an $({\cal F}, {\cal F}_d)$-partition?
\end{quest}

Note that proving $d= 0$ in Question~\ref{q2} would prove Question~\ref{q1}.
The main result of this paper is the following:
\begin{theo} \label{main} 
  Every triangle-free planar graph admits an $({\cal F},{\cal F}_5)$-partition.
\end{theo}

This implies that $d \le 5$ in Question~\ref{q2}. Our proof uses the
discharging method. It is constructive and immediately yields an
algorithm for finding an $({\cal F},{\cal F}_5)$-partition of a
triangle-free planar graph in quadratic time.

Note that Montassier and Ochem~\cite{mo15} proved that not every
triangle-free planar graph can be partitioned into two graphs of
bounded degree (which shows that our result is tight in
some sense).

Finally, we show that if for some $d$, there exists a triangle-free planar graph
that does not admit an $({\cal F},{\cal F}_d)$-partition, then deciding whether a triangle-free planar
graph admits such a partition is NP-complete. That is, if the answer
to Question~\ref{q2} is some $k> 0$, then for all $0\le d< k$, deciding whether a triangle-free
planar graph admits an $({\cal F},{\cal F}_d)$-partition is NP-complete. We prove this by
reduction to \textsc{Planar 3-Sat}.

All presented results on vertex-partition of planar graphs are
summarized in Table~\ref{tab:results}. 

\begin{table}
  \centering
  \begin{center}
    \begin{tabular}{|l|l|l|}
      \hline
      Classes & Vertex-partitions & References \\ \hline
      \multirow{5}{*}{Planar graphs} & $({\cal I},{\cal I},{\cal I},{\cal I})$ & The Four
                                                                                 Color Theorem \cite{appel1,appel2}\\
              & $({\cal I},{\cal F},{\cal F})$& Borodin \cite{Borodin} \\
              & $({\cal F}_2,{\cal F}_2,{\cal F}_2)$& Poh \cite{Poh} \\
              & $({\cal F},{\cal D}_2)$& Thomassen \cite{thomassen1995decomposing} \\
              & $({\cal I},{\cal D}_3)$& Thomassen \cite{thomassen2001decomposing} \\ \hline
      \multirow{4}{*}{Planar graphs with girth 4} & $({\cal I},{\cal I},{\cal I})$& Grötzsch
                                                                                    \cite{grotzsch} \\
              & $({\cal F},{\cal F})$ & Folklore \\
              & $({\cal F}_5,{\cal F})$ & Present paper (Theorem~\ref{main}) \\ 
              & $({\cal I},{\cal F})$ & Open question (Question~\ref{q1})\\ \hline 
      \multirow{1}{*}{Planar graphs with girth 5} & $({\cal I},{\cal F})$& Borodin and
                                                                           Glebov \cite{borodin2001partition} \\ \hline
    \end{tabular}
  \end{center}

  \caption{Known results.}
  \label{tab:results}
\end{table}
\medskip

Theorem~\ref{main} will be proved in Section~\ref{proofmain}.
Section~\ref{complexity} is devoted to complexity results.

\subsection*{Notation} \label{not}
Let $G=(V,E)$ be a plane graph (i.e. planar graph together with its
embedding).

For a set $S \subset V$, let $G - S$ be the graph constructed from $G$ by removing the vertices of $S$ and all the edges incident to some vertex of $S$. If $x \in V$, then we denote $G - \{x\}$ by $G - x$. For a set $S$ of vertices such that $S \cap V = \emptyset$, let $G + S$ be the graph constructed from $G$ by adding the vertices of $S$. If $x \notin V$, then we denote $G + \{x\}$ by $G + x$. For a set $E'$ of pairs of vertices of $G$ such that $E' \cap E = \emptyset$, let $G + E'$ be the graph constructed from $G$ by adding the edges of $E'$. If $e$ is a pair of vertices of $G$ and $e \notin E$, then we denote $G + \{e\}$ by $G + e$. For a set $W \subset V$, we denote by $G[W]$ the subgraph of $G$ induced by $W$.

We call a vertex of degree $k$, at least $k$ and at most $k$, a
\emph{$k$-vertex}, a \emph{$k^+$-vertex} and a \emph{$k^-$-vertex}
respectively, and by extension, for any fixed vertex $v$, we call a
neighbour of $v$ of degree $k$, at least $k$ and at most $k$, a
\emph{$k$-neighbour}, a \emph{$k^+$-neighbour}, and a
\emph{$k^-$-neighbour} of $v$ respectively. When there is some
ambiguity on the graph, we call a neighbour of $v$ in $G$ a
\emph{$G$-neighbour} of $v$. We call a cycle of length $\ell$, at least
$\ell$ and at most $\ell$ a \emph{$\ell$-cycle}, a \emph{$\ell^+$-cycle}, and a
\emph{$\ell^-$-cycle} respectively, and by extension a face of length
$\ell$, at least $\ell$ and at most $\ell$ a \emph{$\ell$-face}, a
\emph{$\ell^+$-face}, and a \emph{$\ell^-$-face} respectively. We say that a
vertex of $G$ is \emph{big} if it is a $8^+$-vertex, and \emph{small}
otherwise. By extension, a big neighbour of a vertex $v$ is a
$8^+$-neighbour of $v$, and a small neighbour of $v$ is a
$7^-$-neighbour of $v$.

Two neighbours $u$ and $w$ of a vertex $v$ are \emph{consecutive} if $uvw$ forms a path on the boundary of a face.








\section{Proof of Theorem \ref{main}} \label{proofmain}

We prove Theorem~\ref{main} by contradiction. Let $G = (V,E)$ be a counter-example to Theorem~\ref{main} of minimum order.

Graph $G$ is connected, otherwise at least one of its connected components would be a counter-example to Theorem~\ref{main}, contradicting the minimality of $G$.

Let us consider any plane embedding of $G$.
Let us prove a series of lemmas on the structure of $G$, that correspond to forbidden configurations in $G$.

\begin{lemm} \label{degge3}
  There are no $2^-$-vertices in $G$.
\end{lemm}

\begin{proof}
  Suppose there is a $2^-$-vertex $v$ in $G$. By minimality of $G$, $G-v$ admits an $({\cal F},{\cal F}_5)$-partition $(F,D)$. If $v$ is a $1^-$-vertex, then $G[F \cup\{v\}] \in {\cal F}$. Suppose $v$ is a $2$-vertex. If both of its neighbours are in $F$, then $G[D \cup \{v\}] \in {\cal F}_5$. Otherwise, $G[F \cup\{v\}] \in {\cal F}$. In all cases, one can obtain an $({\cal F},{\cal F}_5)$-partition of $G$, a contradiction.
\end{proof}

\begin{lemm} \label{3-b}
  Every $3$-vertex in $G$ has at least one big neighbour.
\end{lemm}

\begin{proof}
  Suppose there is a $3$-vertex $v$ in $G$ that has three small neighbours. By minimality of $G$, $G - v$ admits an $({\cal F},{\cal F}_5)$-partition $(F,D)$. If at least two neighbours of $v$ are in $D$, then $G[F \cup\{v\}] \in {\cal F}$. If no neighbour of $v$ is in $D$, then $G[D \cup \{v\}] \in {\cal F}_5$. Suppose exactly one neighbour $u$ of $v$ is in $D$. If at most one of the neighbours of $u$ is in $F$, then $G[F \cup \{u\}] \in {\cal F}$, and $G[D \backslash \{u\} \cup \{v\}] \in {\cal F}_5$. Otherwise, since $u$ is small, at most four of the neighbours of $u$ are in $D$, thus $G[D \cup \{v\}] \in {\cal F}_5$. In all cases, one can obtain an $({\cal F},{\cal F}_5)$-partition of $G$, a contradiction.
\end{proof}

\begin{lemm} \label{4-5star}
  Every $4$-vertex or $5$-vertex in $G$ has at least one $4^+$-neighbour.
\end{lemm}

\begin{proof}
  Suppose there is a $4$-vertex or $5$-vertex $v$ in $G$ that has no $4^+$-neighbour. Let the $u_i$ be the neighbours of $v$, for $i \in \{0,...,3\}$ or $i \in \{0,...,4\}$. Let $G' = G - v - \bigcup_i \{u_i\}$. By minimality of $G$, $G'$ admits an $({\cal F},{\cal F}_5)$-partition $(F,D)$. Add $v$ to $D$, and for all $u_i$, add $u_i$ to $D$ if its two neighbours distinct from $v$ are in $F$, and add $u_i$ to $F$ otherwise. Vertex $v$ has at most five neighbours in $D$, and each of the $u_i$ that is in $D$ has one neighbour in $D$. Each of the $u_i$ that is in $F$ has at most one neighbour in $F$. We have an $({\cal F},{\cal F}_5)$-partition of $G$, a contradiction.
\end{proof}

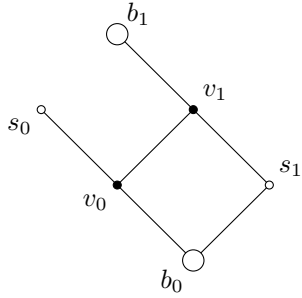
\begin{figure}[h]
  \begin{center}
    \begin{tikzpicture}
      \coordinate (v0) at (0,0) ;
      \coordinate (v1) at (1,1) ;
      \coordinate (b0) at (1,-1) ;
      \coordinate (b1) at (0,2) ;
      \coordinate (s0) at (-1,1) ;
      \coordinate (s1) at (2,0) ;

      \draw (v0) node [below left] {$v_0$} ; 
      \draw (v1) node [above right] {$v_1$} ;
      \draw (b0) node [below left] {$b_0$} ;
      \draw (b1) node [above right] {$b_1$} ;
      \draw (s0) node [below left] {$s_0$} ;
      \draw (s1) node [above right] {$s_1$} ;

      \draw (v0) -- (v1);
      \draw (v0) -- (s0);
      \draw (v0) -- (b0);
      \draw (v1) -- (s1);
      \draw (v1) -- (b1);
      \draw (b0) -- (s1);

      \draw [fill=black] (v0) circle (1.5pt) ; 
      \draw [fill=black] (v1) circle (1.5pt) ;
      \draw [fill=white] (b0) circle (4pt) ;
      \draw [fill=white] (b1) circle (4pt) ;
      \draw [fill=white] (s0) circle (1.5pt) ;
      \draw [fill=white] (s1) circle (1.5pt) ;

    \end{tikzpicture}
  \end{center}
  \caption{The forbidden configuration of Lemma~\ref{config1}. The big vertices are represented with big circles, and the small vertices with small circles. The filled circles represent vertices whose incident edges are all represented.}\label{figconfig1}
\end{figure}

We will need the following observation in the next two lemmas.
\begin{obs} \label{dist3}
  Let $v_0v_1v_2v_3$ be a face of $G$, $u_0$ a neighbour of $v_0$ and $u_1$ a neighbour of $v_1$. Either $u_0$ and $v_2$ are at distance at least $3$, or $u_1$ and $v_3$ are at distance at least $3$.
\end{obs}
By contradiction, suppose that $u_0$ and $v_2$ are at distance at most two, and that $u_1$ and $v_3$ are at distance at most two. Since $G$ is triangle-free, a shortest path from $u_0$ to $v_2$ (resp. from $u_1$ to $v_3$) does not contain any of the $u_i$ and $v_i$ except for its extremities. Then by planarity there exists a vertex $w$ adjacent to $u_0$, $v_2$, $u_1$ and $v_3$. In particular $v_2v_3w$ is a triangle, a contradiction.

\begin{lemm} \label{config1}
  The following configuration does not occur in $G$: two adjacent $3$-vertices $v_0$ and $v_1$ such that for $i \in \{0,1\}$, $v_i$ has a big neighbour $b_i$ and a small neighbour $s_i$, and such that $v_0v_1s_1b_0$ bounds a face of $G$.
\end{lemm}

\begin{proof}
  Suppose such a configuration exists in $G$. See Figure~\ref{figconfig1} for an illustration of this configuration. Observe that all the vertices defined in the statement are distinct (since $G$ is triangle-free). By Observation~\ref{dist3}, either $b_0$ and $b_1$ are at distance at least $3$, or $s_0$ and $s_1$ are at distance at least $3$. For the remaining of the proof, we no longer need the fact that $b_0s_1 \in E(G)$. We forget this assumption, and only remember that either $b_0$ and $b_1$ are at distance at least $3$, or $s_0$ and $s_1$ are at distance at least $3$. This provides some symmetry in the graph.

  Let $G_0 = G - \{v_0,v_1\} + b_0b_1$ and $G_1 = G - \{v_0,v_1\} + s_0s_1$. By what precedes, either $G_0$ or $G_1$ is triangle-free, thus there exists a $j$ such that $G_j$ is a triangle-free planar graph. By minimality of $G$, $G_j$ admits an $({\cal F},{\cal F}_5)$-partition $(F,D)$. 

  Let us first prove that if we do not have $b_0$ and $b_1$ in $D$, and $s_0$ and $s_1$ in $F$, then the conditions $G[F] \in {\cal F}$ and $G[D] \in {\cal F}_5$ lead to a contradiction. We will see that we can always extend the $({\cal F},{\cal F}_5)$-partition of $G_j$ to $G$.

  \begin{itemize}
  \item
    If at least three of the $b_i$ and $s_i$ are in $D$, then $G[F \cup \{v_0,v_1\}] \in {\cal F}$. 

  \item If all of the $b_i$ and $s_i$ are in $F$, then $G[D \cup \{v_0,v_1\}] \in {\cal F}_5$.

  \item
    Suppose now that exactly three of the $b_i$ and $s_i$ are in $F$. W.l.o.g., $b_0 \in D$ or $s_0 \in D$. We have $G[F \cup \{v_0\}] \in {\cal F}$ and $G[D \cup \{v_1\}] \in {\cal F}_5$.

  \item
    Suppose now that exactly two of the $b_i$ and $s_i$ are in $F$. If $b_0$ and $s_0$ are in $F$ (resp. $b_1$ and $s_1$ are in $F$), then $G[D \cup \{v_0\}] \in {\cal F}_5$ and $G[F \cup \{v_1\}] \in {\cal F}$ (resp. $G[F \cup \{v_0\}] \in {\cal F}$ and $G[D \cup \{v_1\}] \in {\cal F}_5$).

    Now w.l.o.g. $b_0 \in F$ and $s_0 \in D$. If $s_0$ has at most one $G$-neighbour in $F$, then  $G[F \cup \{s_0\}] \in {\cal F}$, we can replace $F$ by $F \cup \{s_0\}$ and $D$ by $D \backslash \{s_0\}$, and we fall into a previous case. We can thus assume that $s_0$ has at least two of its $G$-neighbours in $F$, and thus it has at most four of its $G$-neighbours in $D$. Therefore $G[D \cup \{v_0\}] \in {\cal F}_5$, and $G[F \cup \{v_1\}] \in {\cal F}$.
  \end{itemize}
  In all cases, $G$ has an $({\cal F},{\cal F}_5)$-partition, a contradiction.
  \medskip

  Remains the case where $b_0$ and $b_1$ are in $D$, and $s_0$ and $s_1$ are in $F$. In the case where we added the edge $b_0b_1$ (i.e. the case $j = 0$), we have $G[D \cup \{v_0,v_1\}] \in {\cal F}_5$, since $G[D \cup \{v_0,v_1\}]$ is equal to $G_0[D]$ where an edge is subdivided twice. Similarily, in the case where we added the edge $s_0s_1$ (i.e. the case $j = 1$), we have $G[F \cup \{v_0,v_1\}] \in {\cal F}$, since $G[F \cup \{v_0,v_1\}]$ is equal to $G_0[F]$ where an edge is subdivided twice. Again, $G$ has an $({\cal F},{\cal F}_5)$-partition, a contradiction.
\end{proof}

\begin{figure}[h]
  \begin{center}
    \begin{tikzpicture}
      \coordinate (v0) at (0,0);
      \coordinate (v1) at (1,1);
      \coordinate (s1) at (2,0);
      \coordinate (b0) at (1,-1);
      \coordinate (s0) at (-1,1);
      \coordinate (w0) at (0,2);
      \coordinate (w1) at (2,2);

      \draw (v0) node [below left] {$v_0$} ;
      \draw (1.1,1) node [right] {$v_1$} ;
      \draw (s1) node [below right] {$s_1$} ;
      \draw (1,-1.1) node [below] {$b$} ;
      \draw (s0) node [left] {$s_0$} ;
      \draw (w0) node [above left] {$w_0$} ;
      \draw (w1) node [above right] {$w_1$} ;

      \draw (v0) -- (v1);
      \draw (v0) -- (b0);
      \draw (v0) -- (s0);
      \draw (v1) -- (s1);
      \draw (v1) -- (w0);
      \draw (v1) -- (w1);
      \draw (s1) -- (b0);
      \draw (s1) -- ++(0.5,0.5);

      \draw [fill=black](v0) circle (1.5pt) ;
      \draw [fill=black](v1) circle (1.5pt) ;
      \draw [fill=black](s1) circle (1.5pt) ;
      \draw [fill=white](b0) circle (4pt) ;
      \draw [fill=white](s0) circle (1.5pt) ;
      \draw [fill=white](w0) circle (1.5pt) ;
      \draw [fill=white](w1) circle (1.5pt) ;

    \end{tikzpicture}
  \end{center}
  \caption{The forbidden configuration of Lemma~\ref{config2}.}\label{figconfig2}
\end{figure}
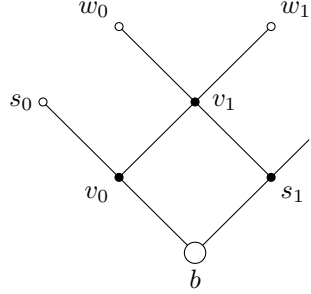

\begin{lemm} \label{config2}
  The following configuration does not occur in $G$: a $3$-vertex $v_0$ adjacent to a $4$-vertex $v_1$ such that $v_0$ has a big neighbour $b$ and a small neighbour $s_0$, and $v_1$ has three other small neighbours $s_1$, $w_0$, and $w_1$ such that $v_0v_1s_1b$ bounds a face of $G$ and $s_1$ has degree $3$.
\end{lemm}

\begin{proof}
  Suppose such a configuration exists in $G$. See Figure~\ref{figconfig2} for an illustration of this configuration. Observe that all the vertices defined in the statement are distinct (since $G$ is triangle-free). By Observation~\ref{dist3}, either $b$ and $w_0$ are at distance at least $3$, or $s_0$ and $s_1$ are at distance at least $3$. Let $G_0 = G - \{v_0,v_1\} + bw_0$ and $G_1 = G - \{v_0,v_1\} + s_0s_1$. By what precedes, either $G_0$ or $G_1$ is triangle-free, thus there exists a $j$ such that $G_j$ is a triangle-free planar graph. By minimality of $G$, $G_j$ has an $({\cal F},{\cal F}_5)$-partition $(F,D)$.

  Let us first prove that except in the case where $\{b,w_0,w_1\} \subset D$ and $\{s_0,s_1\} \subset F$, the conditions $G[F] \in {\cal F}$ and $G[D] \in {\cal F}_5$ lead to a contradiction. We will see that we can always extend the $({\cal F},{\cal F}_5)$-partition of $G_j$ to $G$.

  If at least four among the $w_i$, the $s_i$ and $b$ are in $D$, then $G[F \cup \{v_0,v_1\}] \in {\cal F}$.

  Suppose now that at most three among the $w_i$, the $s_i$ and $b$ are in $D$.
  Suppose $x \in \{b,s_0,s_1,w_0,w_1\}$ is in $D$. If $x$ has at most one $G$-neighbour in $F$, then $G[F \cup \{x\}] \in {\cal F}$, and we could consider $F \cup \{x\}$ instead of $F$ and $D \backslash \{x\}$ instead of $D$. Note that this cannot lead to the case we excluded ($\{b,w_0,w_1\} \subset D$ and $\{s_0,s_1\} \subset F$) unless at least four among the $w_i$, the $s_i$ and $b$ are in $D$. Thus we can assume that for any $x$ among the $w_i$ and $s_i$ such that $x \in D$, $x$ has at most four $G$-neighbours in $D$, and thus adding one neighbour of $x$ in $D$ cannot cause $x$ to have at least six neighbours in $D$. We consider two cases according to $b$:

  \begin{itemize}

  \item Suppose $b \in F$. If at least three of the $w_i$ and $s_i$ are in $F$, then $G[D \cup \{v_0,v_1\}] \in {\cal F}_5$. 

    If at least two among the $w_i$ and $s_1$ are in $D$, then $G[F \cup \{v_1\}] \in {\cal F}$ and $G[D \cup \{v_0\}] \in {\cal F}_5$. Else, at least two among the $w_i$ and $s_1$ are in $F$, and we may assume that $s_0$ is in $D$ (otherwise we fall into a previous case), so $G[D \cup \{v_1\}] \in {\cal F}_5$ and $G[F \cup \{v_0\}] \in {\cal F}$.

  \item Suppose now that $b \in D$. As $s_1$ has degree $3$, it has at most one $G$-neighbour in $F$, and thus as previously we could consider $F \cup \{s_1\}$ instead of $F$ and $D \backslash \{s_1\}$ instead of $D$. Again, this cannot lead to the case we excluded ($\{b,w_0,w_1\} \subset D$ and $\{s_0,s_1\} \subset F$) unless at least four among the $w_i$, the $s_i$ and $b$ are in $D$. Therefore we can assume that $s_1 \in F$. The $w_i$ are not both in $D$ (otherwise we fall into the case we excluded). We have $G[D \cup \{v_1\}] \in {\cal F}_5$ and $G[F \cup \{v_0\}] \in {\cal F}$.
  \end{itemize}

  In all cases, $G$ has an $({\cal F},{\cal F}_5)$-partition, a contradiction.

  Remains the case $\{b,w_0,w_1\} \subset D$ and $\{s_0,s_1\} \subset F$. In the case where we added the edge $bw_0$ (i.e. the case $j = 0$), $b$ has at most five $G_0$-neighbours in $D$, and thus at most four $G$-neighbours in $D$, so $G[D \cup \{v_0\}] \in {\cal F}_5$, and $G[F \cup \{v_1\}] \in {\cal F}$. In the case where we added the edge $s_0s_1$ (i.e. the case $j = 1$), we have $G[F \cup \{v_0,v_1\}] \in {\cal F}$, since $G[F \cup \{v_0,v_1\}]$ is equal to $G_0[F]$ where an edge is subdivided twice. Again, $G$ has an $({\cal F},{\cal F}_5)$-partition, a contradiction.
\end{proof}

\begin{figure}[h]
  \begin{center}
    \begin{tikzpicture}
      \coordinate (b0) at (0,0);
      \coordinate (w0) at (1,-1);
      \coordinate (v0) at (1,1);
      \coordinate (v1) at (2,0);
      \coordinate (v2) at (3,1);
      \coordinate (v3) at (2,2);
      \coordinate (b1) at (3,-1);

      \draw (-0.1,0) node [left] {$b_0$} ;
      \draw (w0) node [left] {$w_0$} ;
      \draw (v0) node [left] {$v_0$} ;
      \draw (v1) node [right] {$v_1$} ;
      \draw (v2) node [right] {$v_2$} ;
      \draw (v3) node [above] {$v_3$} ;
      \draw (3.1,-1) node [right] {$b_1$} ;

      \draw (b0) -- (w0);
      \draw (v1) -- (w0);
      \draw (b0) -- (v0);
      \draw (b1) -- (v1);
      \draw (v0) -- (v1);
      \draw (v1) -- (v2);
      \draw (v2) -- (v3);
      \draw (v0) -- (v3);
      \draw (w0) -- ++ (0.5, -0.5);

      \draw [fill=white](b0) circle (4pt) ;
      \draw [fill=black](w0) circle (1.5pt) ;
      \draw [fill=black](v0) circle (1.5pt) ;
      \draw [fill=black](v1) circle (1.5pt) ;
      \draw [fill=white](v2) circle (1.5pt) ;
      \draw [fill=white](v3) circle (1.5pt) ;
      \draw [fill=white](b1) circle (4pt) ;

    \end{tikzpicture}
  \end{center}
  \caption{Configuration~\ref{config3}.\label{figconfig3}}
\end{figure}

We define a specific configuration:

\begin{config} \label{config3}
  Two $4$-faces $b_0v_0v_1w_0$ and $v_0v_1v_2v_3$, such that $b_0$ is a big vertex, $v_0$ and $w_0$ are $3$-vertices, $v_1$ is a $4$-vertex, $v_2$ and $v_3$ are small vertices, and the fourth neighbour of $v_1$, say $b_1$, is a big vertex.
  See Figure~\ref{figconfig3} for an illustration of this configuration.
\end{config}

\begin{figure}[h]
  \begin{center}
    \begin{tikzpicture}
      \coordinate (b0) at (0,0);
      \coordinate (w0) at (1,-1);
      \coordinate (w0') at (2,-2);
      \coordinate (v0) at (1,1);
      \coordinate (v1) at (2,0);
      \coordinate (v2) at (3,1);
      \coordinate (v2') at (4,2);
      \coordinate (v3) at (2,2);
      \coordinate (b1) at (3,-1);
      \coordinate (w1) at (4,0);
      \coordinate (w1') at (5,1);

      \draw (-0.1,0) node [left] {$b_0$} ;
      \draw (w0) node [left] {$w_0$} ;
      \draw (v0) node [left] {$v_0$} ;
      \draw (v1) node [right] {$v_1$} ;
      \draw (v2) node [right] {$v_2$} ;
      \draw (v3) node [above] {$v_3$} ;
      \draw (3.1,-1) node [right] {$b_1$} ;
      \draw (w1) node [right] {$w_1$} ;

      \draw (b0) -- (w0);
      \draw (v1) -- (w0);
      \draw (b0) -- (v0);
      \draw (b1) -- (v1);
      \draw (b1) -- (w1);
      \draw (v0) -- (v1);
      \draw (v1) -- (v2);
      \draw (v2) -- (v3);
      \draw (v0) -- (v3);
      \draw (w0) -- (w0');
      \draw (v2) -- (v2');
      \draw (v2) -- (w1);
      \draw (w1) -- (w1');

      \draw [fill=white](b0) circle (4pt) ;
      \draw [fill=black](w0) circle (1.5pt) ;
      \draw [fill=white](w0') circle (1.5pt) ;
      \draw [fill=black](v0) circle (1.5pt) ;
      \draw [fill=black](v1) circle (1.5pt) ;
      \draw [fill=black](v2) circle (1.5pt) ;
      \draw [fill=white](v2') circle (1.5pt) ;
      \draw [fill=white](v3) circle (1.5pt) ;
      \draw [fill=white](b1) circle (4pt) ;
      \draw [fill=black](w1) circle (1.5pt) ;
      \draw [fill=white](w1') circle (1.5pt) ;

    \end{tikzpicture}
  \end{center}
  \caption{The forbidden configuration of Lemma~\ref{config4}.\label{figconfig4}}
\end{figure}

\begin{lemm} \label{config4}
  The following configuration is forbidden: Configuration~\ref{config3} with the added condition that there is a $4$-face $b_1v_1v_2w_1$ with $w_1$ a $3$-vertex, $v_2$ a $4$-vertex, and the fourth neighbour of $v_2$, the third neighbour of $w_1$, and the third neighbour of $w_0$ are small vertices.
\end{lemm}

\begin{proof}
  Suppose such a configuration exists in $G$. See Figure~\ref{figconfig4} for an illustration of this configuration. Observe that all the vertices named in the statement are distinct since $G$ is triangle-free and $w_1$ is a small vertex whereas $b_0$ is a big one.

  Let us prove that either $b_0$ and $b_1$ are at distance at least $3$, or $w_0$ and $w_1$, and $w_0$ and $v_3$ are at distance at least $3$. By contradiction, suppose that $b_0$ and $b_1$ are at distance at most two, and that either $w_0$ and $w_1$ are at distance at most two, or $w_0$ and $v_3$ are at distance at most $2$. Since $G$ is triangle-free, a shortest path from $b_0$ to $b_1$, from $w_0$ to $w_1$ or from $w_0$ to $v_3$ does not go through any of the vertices defined in the statement. Then by planarity there exists a vertex $w$ adjacent to $b_0$, $b_1$, $w_0$ and either $w_1$ or $v_3$. In particular $b_0w_0w$ is a triangle, a contradiction.

  Let $G_0 = G - \{v_0,v_1\} + b_0b_1$ and $G_1 = G - \{v_0,v_1\} + w_0w_1 + w_0v_3$. By what precedes, either $G_0$ or $G_1$ is triangle-free, thus there exists a $j$ such that $G_j$ is a triangle-free planar graph. By minimality of $G$, $G_j$ has an $({\cal F},{\cal F}_5)$-partition $(F,D)$.

  Let $s_0$ be the third neighbour of $w_0$, $s_1$ be the third neighbour of $w_1$ and $s_2$ be the fourth neighbour of $v_2$. They are all small vertices, but there may be some that are equal between themselves, or equal to some vertices we defined previously. However, if one of the $s_i$ is in $\{v_0,v_1,v_2,w_0,w_1\}$, then this $s_i$ is a $4^-$-vertex in $G$ (and in particular it has at most $4$ neighbours in $D$).

  Suppose first that $b_0$ and $b_1$ are both in $D$.
  \begin{enumerate}
  \item \label{1stpoint} Suppose $w_0$ is in $D$. Here we only consider $(F,D)$ as an $({\cal F},{\cal F}_5)$-partition of $G - \{v_0,v_1\}$.

    If $v_3$ is also in $D$, then adding $v_0$ and $v_1$ to $F$ leads to an $({\cal F},{\cal F}_5)$-partition of $G$. Suppose $v_3$ is in $F$. We show now that we can assume that $v_2$ is in $D$. By contradiction, suppose $v_2$ is in $F$. We remove $v_2$ from $F$.
    
    Observe that we can assume that $v_2$ has no $G$-neighbour in $D$ with five $G$-neighbours in $D$. Indeed, suppose $v_2$ has a $G$-neighbour in $D$ with five $G$-neighbours in $D$. This $G$-neighbour is a $5^+$-vertex, so it is $s_2$. Moreover, $s_2$ is not equal to $v_3$ (because $v_3$ is in $F$), and is not equal to any of the other vertices named in the statement (because of the degree conditions). As $s_2$ is a small $D$-vertex, has at least five $G$-neighbours in $D$ and is adjacent to $v_2$ that is neither in $F$ nor in $D$, $s_2$ has at most one neighbour in $F$. Therefore we can put $s_2$ in $F$.

    Observe that we can assume that $v_2$ has at most one $G$-neighbour in $D$. Suppose $v_2$ has two $G$-neighbours in $D$. These $G$-neighbours are $s_2$ and $w_1$. Vertex $w_1$ has at most one neighbour in $F$, that is $s_1$, so we can put $w_1$ in $F$.

    Now $v_2$ has at most one $G$-neighbour in $D$, and no $G$-neighbour of $v_2$ in $D$ has five $G$-neighbours in $D$, so we can put $v_2$ in $D$. Therefore we can always assume that $v_2$ is in $D$. Note that we do not need to change where $s_2$ is in the partition if it is equal to one of the vertices named in the statement. Adding $v_0$ and $v_1$ to $F$ leads to an $({\cal F},{\cal F}_5)$-partition of $G$. 

  \item Suppose $w_0$ is in $F$, $v_3$ is in $D$ and $w_1$ is in $D$. If $s_2$ is in $D$, then putting $v_0$, $v_1$ and $v_2$ in $F$ leads to an $({\cal F},{\cal F}_5)$-partition of $G$. Suppose $s_2$ is in $F$. We put $v_0$, $v_1$ and $w_1$ in $F$, and $v_2$ in $D$. If this increases the number of $G$-neighbours of $v_3$ in $D$ above five, then since $v_3$ is small, $v_3$ has at most one neighbour in $F$, which is $v_0$, and we put $v_3$ in $F$. This leads to an $({\cal F},{\cal F}_5)$-partition of $G$. 

  \item Suppose $w_0$ is in $F$, $v_3$ is in $D$ and $w_1$ is in $F$. Suppose $s_2$ is in $F$. We put $v_0$ and $v_1$ in $F$, and $v_2$ in $D$. If this increases the number of $G$-neighbours of $v_3$ in $D$ above five, then since $v_3$ is small, $v_3$ has at most one neighbour in $F$, which is $v_0$, and we put $v_3$ in $F$. This leads to an $({\cal F},{\cal F}_5)$-partition of $G$. Suppose $s_2$ is in $D$. If $v_2$ is not in $F$, we may put it in $F$, since it has only one $G_j$-neighbour in $F$, that is $w_1$. Therefore we can assume that $v_2$ is in $F$. If $j = 0$, then $b_1$ has at most $4$ $G$-neighbours in $D$ (since it has at most $5$ such $G_0$-neighbours), so adding $v_0$ to $F$ and $v_1$ to $D$ leads to an $({\cal F},{\cal F}_5)$-partition of $G$. If $j = 1$, then adding $v_0$ and $v_1$ to $F$ leads to an $({\cal F},{\cal F}_5)$-partition of $G$.

  \item Suppose $w_0$ is in $F$ and $v_3$ is in $F$. Suppose $j = 0$. The vertex $b_0$ has at most $4$ $G$-neighbours in $D$ (since it has at most $5$ such $G_0$-neighbours), so we can add $v_0$ to $D$. If $v_2$ is in $D$, then adding $v_1$ to $F$ leads to an $({\cal F},{\cal F}_5)$-partition of $G$. If $v_2$ is in $F$, then adding $v_1$ to $D$ makes $G[D]$ equal to $G_0[D]$ with an edge subdivided twice, and this leads to an $({\cal F},{\cal F}_5)$-partition of $G$. Suppose $j = 1$. Here we only consider $(F,D)$ as an $({\cal F},{\cal F}_5)$-partition of $G - \{v_0,v_1\} + w_0v_3$. As in \ref{1stpoint}, we can suppose, up to changing where $s_2$ and $w_1$ are in the partition, that $v_2$ is in $D$. Note that if $s_2$ is equal to one of the vertices named in the statement, we do not need to move $s_2$ in the partition. Adding $v_0$ and $v_1$ to $F$ leads to an $({\cal F},{\cal F}_5)$-partition of $G$.
  \end{enumerate}

  Now we may assume that at least one of $b_0$ and $b_1$ is in $F$. From now on we only consider $(F,D)$ as an $({\cal F},{\cal F}_5)$-partition of $G - \{v_0,v_1\}$.

  \begin{itemize}
  \item Suppose $b_0$ is in $F$ and $b_1$ is in $D$. In that case we put $v_0$ and $w_0$ in $D$, and $v_1$ in $F$. Adding $v_0$ in $D$ (resp. $w_0$ in $D$) may violate the degree condition of $G[D]$ ; however, if it happens, one can put $v_3$ (resp. $s_0$) in $F$. In any case, we obtain an $({\cal F},{\cal F}_5)$-partition of $G$.

  \item Suppose $b_0$ is in $D$ and $b_1$ is in $F$. If at least one of $w_0$ and $v_2$ is in $F$, then adding $v_0$ in $F$ and $v_1$ in $D$ leads to an $({\cal F},{\cal F}_5)$-partition of $G$. Assume $w_0$ and $v_2$ are both in $D$. If $v_3$ is in $D$, then adding $v_0$ and $v_1$ in $F$ leads to an $({\cal F},{\cal F}_5)$-partition of $G$. Assume $v_3$ is in $F$. We consider three cases:

    \begin{itemize}
    \item Suppose $s_2$ and $w_1$ are in $F$. Adding $v_0$ in $F$ and $v_1$ in $D$ leads to an $({\cal F},{\cal F}_5)$-partition of $G$.

    \item Suppose $s_2$ is in $F$ and $w_1$ is in $D$. If $s_1$ is in $D$, then we can put $w_1$ in $F$ and we fall into the previous case. If $s_1$ is in $F$, then adding $v_0$ in $F$ and $v_1$ in $D$ leads to an $({\cal F},{\cal F}_5)$-partition of $G$.

    \item Suppose $s_2$ is in $D$. If $s_1$ is in $D$ and has five $G$-neighbours in $D$ distinct from $w_1$, then as $s_1$ is small, it is distinct from all the vertices named in the statement, and we can put it in $F$. Therefore we can put $w_1$ in $D$ and $v_2$ in $F$. We fall into a previous case (at least one of $w_0$ and $v_2$ is in $F$).
    \end{itemize}

  \item Suppose $b_0$ and $b_1$ are in $F$. If $s_0$ is in $D$ and has five $G$-neighbours in $D$ distinct from $w_0$, then as $s_0$ is small, it is distinct from all the vertices named in the statement aside from $v_3$, and we can put it in $F$. Therefore we can put $w_0$ in $D$. We consider the following cases:

    \begin{itemize}
    \item If $v_2$ and $v_3$ are in $F$, then adding $v_0$ and $v_1$ to $D$ leads to an $({\cal F},{\cal F}_5)$-partition of $G$.

    \item If $v_2$ is in $F$ and $v_3$ is in $D$, then adding $v_0$ to $F$ and $v_1$ to $D$ leads to an $({\cal F},{\cal F}_5)$-partition of $G$.

    \item If $v_2$ is in $D$ and $v_3$ is in $F$, then adding $v_0$ to $D$ and $v_1$ to $F$ leads to an $({\cal F},{\cal F}_5)$-partition of $G$.

    \item If $v_2$ and $v_3$ are in $D$, then adding $v_0$ to $D$ and $v_1$ to $F$ leads to an $({\cal F},{\cal F}_5)$-partition of $G$. Adding $v_0$ to $D$ may violate the degree condition of $G[D]$, but in that case we can put $v_3$ in $F$.
    \end{itemize}
  \end{itemize}

\end{proof}

We now apply a discharging procedure: first, for all $j$, every $j$-vertex $v$ has a charge equal to $c_0(v) = j - 4$, and every $j$-face $f$ has a charge equal to $c_0(f) = j - 4$. By Euler's formula, the total charge is negative (equal to $-8$). Observe that, since $G$ is triangle-free, every face has a non-negative initial charge, and by Lemma~\ref{degge3}, the vertices that have negative initial charges are exactly the $3$-vertices of $G$, and they have an initial charge of $-1$. Here is our discharging procedure:
\medskip


\textbf{Discharging procedure:} 

\begin{itemize}
\item
  \emph{Step 1}: Every big vertex gives $\frac{1}{2}$ to each of its small neighbours. Furthermore, for every $4$-face $uvwx$ where $u$ and $v$ are big, and $w$ and $x$ are small, $v$ gives $\frac{1}{4}$ to $x$ (and $u$ gives $\frac{1}{4}$ to $w$).

\item
  \emph{Step 2}: Consider a $4$-vertex $v$ that does not correspond to $v_1$ in Configuration~\ref{config3}. Vertex $v$ gives $\frac{1}{4}$ to each of its small neighbours that are consecutive (as neighbours of $v$) to exactly one big vertex, and $\frac{1}{2}$ to each of its small neighbours that are consecutive (as neighbours of $v$) to two big vertices. 

  Consider the case where $v$ corresponds to $v_1$ in Configuration~\ref{config3}. We use the notations of Configuration~\ref{config3}. If $w_0$ has two big neighbours, then $v_1$ gives $\frac{1}{4}$ to $v_0$ and $\frac{1}{4}$ to $v_2$. Otherwise, it gives $\frac{1}{4}$ to $w_0$ and $\frac{1}{4}$ to $v_0$.

  Every small $5^+$-vertex that has a big neighbour gives $\frac{1}{4}$ to each of its small neighbours, and an additional $\frac{1}{4}$ for each that is consecutive (as neighbours of $v$) to at least one big vertex. Every small $5^+$-vertex that has no big neighbour gives $\frac{1}{4}$ to each of its $3$-neighbours.

\item
  \emph{Step 3}: For every $4$-face $uvwx$, with $u$ a big vertex, $v$ a $3$-vertex, $w$ a $4$-vertex, and $x$ a small vertex such that $x$ gave charge to $w$ in Step $2$, $w$ gives $\frac{1}{4}$ to $v$.

\item
  \emph{Step 4}: Every $5^+$-face that has a big vertex in its boundary gives $\frac{1}{4}$ to each of the small vertices in its boundary. Every $5^+$-face that has no big vertex in its boundary gives $\frac{1}{5}$ to each of the vertices in its boundary.

\item
  \emph{Step 5}: For every $4$-face $uvwx$, with $u$ a big vertex, $v$ a $3$-vertex, $w$ a $4$-vertex and $x$ a $3$-vertex such that the other face that has $vw$ in its boundary is a $5^+$-face, $w$ gives $\frac{1}{5}$ to $v$.
\end{itemize}

For every vertex or face $x$ of $G$, for every $i \in \{1,2,3,4,5\}$, let $c_i(x)$ be the charge of $x$ at the end of Step $i$. Observe that during the procedure, no charges are created and no charges disappear; hence the total charge is kept fixed.

We now prove that every vertex and every face has a non-negative charge at the end of the procedure.
That leads to the following contradiction:
$$0 \le \sum_{x \in V(G) \cup F(G)}c_5(x) = \sum_{x \in V(G) \cup F(G)}c_0(x) = -8$$

\begin{lemm} \label{faces}
  Every face has non-negative charge at the end of the procedure.
\end{lemm}

\begin{proof}
  At the beginning of the procedure, for every $j$-face $f$ we have $c_0(f) = j-4 \ge 0$ (as $j \ge 4$). The procedure does not involve $4$-faces. Hence if $j = 4$, then $c_5(f) = c_0(f) = 0$. If $j = 5$, then $f$ gives at most four times $\frac{1}{4}$ if it is incident to a big vertex and at most five times $\frac{1}{5}$ otherwise in Step $4$. It follows that $c_5(f) \ge 0$. If $j \ge 6$, then $f$ can give $\frac{1}{3}$ to each of its incident vertices (and so $\frac{1}{4}$ or $\frac{1}{5}$) during Step $4$, and $c_5(f) \ge j-4-\frac{j}{3} \ge 0$.
\end{proof}

\begin{lemm} \label{4+vertices}
  A $4^+$-vertex never has negative charge.
\end{lemm}

\begin{proof}
  Consider a $j$-vertex $z$ with $j \ge 4$. At the beginning, $c_0(z) = j-4 \ge 0$. We will show that $c_i(z) \ge 0$ for $i = 1,...,5$.

  \begin{itemize}
  \item
    Suppose $z$ is a big vertex. Such a vertex only loses charge in Step $1$. Since $j \ge 8$, we have $c_0(z) \ge \frac{j}{2}$. In Step $1$, vertex $z$ loses $\frac{1}{2}$ for each of its small neighbours, and at most $\frac{1}{2}$ for each of its big neighbours. Therefore it has more charge than what it gives, and thus it keeps a non-negative charge. 

  \item
    Suppose $z$ is a small $5^+$-vertex. It does not lose charge in Steps $1$, $3$, $4$ and $5$.

    Suppose $z$ has a big neighbour. It has at most $j-1$ small neighbours, and it has charge at least $\frac{1}{4}(j-1)$ at the beginning of the procedure, since $j \ge 5$. Moreover, it receives $\frac{1}{2}$ from each of its big neighbours in Step $1$. Therefore it does not give more charge that it has in Step~$2$.

    Suppose now that $z$ has no big neighbour. If $z$ is a $5$-vertex, then by Lemma~\ref{4-5star}, it has at most four $3$-vertices, and $c_2(z)\ge1-4\frac{1}{4}\ge 0$. If $z$ is a $6^+$-vertex, then $c_2(z)\ge j-4-j\frac{1}{4}\ge 0$.

  \item
    Suppose $z$ is a $4$-vertex. It does not lose charge in Steps $1$ and $4$. Suppose $z$ gives charge in Step $2$. Consider first that $z$ does not correspond to $v_1$ in Configuration~\ref{config3}. If $z$ is adjacent to a small vertex that is consecutive (as a neighbour of $z$) to two big neighbours, then $z$ gives at most twice $\frac{1}{2}$ in Step $2$ and received twice $\frac{1}{2}$ in Step $1$; hence $c_2(z) \ge 0$. Otherwise, $z$ gives at most twice $\frac{1}{4}$ in Step $2$, and received at least once $\frac{1}{2}$ in Step $1$; hence $c_2(z) \ge 0$. Let us now consider the case where $z$ corresponds to $v_1$ in Configuration~\ref{config3}. The vertex $z$ has a big neighbour that gave $\frac{1}{2}$ to $z$ in Step $1$, and $z$ gives $\frac{1}{4}$ to two of its neighbours in Step $2$. Therefore $z$ received in Step $1$ at least as much as what it gives in Step $2$.

    \begin{figure}[h]
      \begin{center}
        \begin{tikzpicture}
          \coordinate (u) at (0,0) ;
          \coordinate (x) at (2,0) ;
          \coordinate (v) at (0,-2) ;
          \coordinate (z) at (2,-2) ;

          \draw (u) node [above left] {$u$} ; 
          \draw (v) node [below left] {$v$} ; 
          \draw (z) node [below right] {$z$} ; 
          \draw (x) node [above right] {$x$} ; 

          \draw (u) -- (v);
          \draw (v) -- (z);
          \draw (z) -- (x);
          \draw (u) -- (x);
          \draw (v) -- (0,-2.65);
          \draw (z) -- (2.65,-2);
          \draw (z) -- (2,-2.65);

          \draw [fill=black] (v) circle (1.5pt) ; 
          \draw [fill=black] (z) circle (1.5pt) ;
          \draw [fill=white] (u) circle (4pt) ;
          \draw [fill=white] (x) circle (1.5pt) ;

        \end{tikzpicture}~~~~~~
        \begin{tikzpicture}
          \coordinate (u) at (0,0) ;
          \coordinate (u') at (4,0) ;
          \coordinate (x) at (2,0) ;
          \coordinate (v) at (0,-2) ;
          \coordinate (v') at (4,-2) ;
          \coordinate (z) at (2,-2) ;

          \draw (u) node [above left] {$u$} ; 
          \draw (v) node [below left] {$v$} ; 
          \draw (u') node [above left] {$u'$} ; 
          \draw (v') node [below left] {$v'$} ; 
          \draw (z) node [below right] {$z$} ; 
          \draw (x) node [above right] {$x$} ; 

          \draw (u) -- (v);
          \draw (u') -- (v');
          \draw (v) -- (z);
          \draw (v') -- (z);
          \draw (z) -- (x);
          \draw (u) -- (x);
          \draw (u') -- (x);
          \draw (v) -- (0,-2.65);
          \draw (v') -- (4,-2.65);
          \draw (z) -- (2,-2.65);

          \draw [fill=black] (v) circle (1.5pt) ;
          \draw [fill=black] (v') circle (1.5pt) ; 
          \draw [fill=black] (z) circle (1.5pt) ;
          \draw [fill=white] (u) circle (4pt) ;
          \draw [fill=white] (u') circle (4pt) ;
          \draw [fill=white] (x) circle (1.5pt) ;

        \end{tikzpicture}
      \end{center}
      \caption{Some configurations that appear in Lemma~\ref{4+vertices}.}\label{fig4+vertices}
    \end{figure}
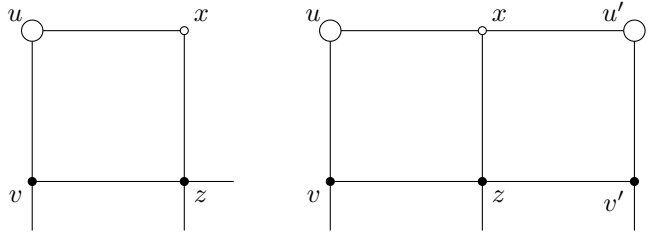

    Suppose $z$ gives charge in Step $3$. There is a $4$-face $uvzx$ with $u$ a big vertex, $v$ a $3$-vertex, and $x$ a small vertex such that $x$ gave charge to $z$ in Step $2$. Suppose $z$ is consecutive to exactly one big vertex (as neighbours of $x$). The vertex $x$ gave at least $\frac{1}{4}$ to $z$ in Step $2$, and there is exactly one such face with the same $z$ and $x$ (i.e. there is no pair $(u',v')$ distinct from $(u,v)$ that verifies the properties we stated for $(u,v)$)(see Figure~\ref{fig4+vertices}, left). Therefore $z$ can give $\frac{1}{4}$ to $v$ in Step $3$. Suppose $z$ is consecutive to exactly two big vertices (as neighbours of $x$). The vertex $x$ gave $\frac{1}{2}$ to $z$ in Step $2$, and there are at most two such faces with the same $z$ and $x$ (i.e. there is at most one pair $(u',v')$ distinct from $(u,v)$ that verifies the properties we stated for $(u,v)$) (see Figure~\ref{fig4+vertices}, right). Therefore $z$ can give $\frac{1}{4}$ to each of the corresponding $v$'s in Step $3$. Therefore $z$ received in Step $2$ at least as much as what it gives in Step $3$.

    Suppose $z$ gives charge in Step $5$. There is a $4$-face $uvzx$, with $u$ a big vertex, $v$ a $3$-vertex, and $x$ a $3$-vertex such that the other face, say $f$, that has $vz$ in its boundary is a $5^+$-face. Vertex $z$ received at least $\frac{1}{5}$ from $f$ in Step $4$, and it gives $\frac{1}{5}$ to $v$. There is a problem only if there is another $4$-face $u'v'zx'$, such that $vzv'$ is on the boundary of $f$, $u'$ is a big vertex, and $x'$ and $v'$ are $3$-vertices. But then $z$ would have four $3$-neighbours, contradicting Lemma~\ref{4-5star}. Therefore $z$ received in Step $4$ at least as much as what it gives in Step $5$.

  \end{itemize}

  In all cases, $z$ never has negative charge.
\end{proof}

\begin{lemm}\label{3vertices}
  At the end of the procedure, every $3$-vertex has non-negative charge.
\end{lemm}

\begin{proof}
  Let $z$ be a $3$-vertex. It never loses charge in the procedure, so we only need to prove that it received at least $1$ over the whole procedure. Assume by contradiction that it received less than that.

  By Lemma~\ref{3-b}, vertex $z$ has at least one big neighbour $b$. Let $x_0$ and $x_1$ be its two other neighbours. Vertex $b$ gives $\frac{1}{2}$ to $z$ in Step $1$, so $z$ only needs to receive $\frac{1}{2}$ from $x_0$, $x_1$, and its surrounding faces. In particular, if one of the $x_i$ is a big vertex, then it gives $\frac{1}{2}$ to $z$ in Step $1$, and $z$ receives all the charge it needs, a contradiction. Therefore $x_0$ and $x_1$ are small vertices.

  Let $f$ be the face that contains $x_0zx_1$ in its boundary, $f_0$ be the face that contains $x_0zb$ in its boundary and $f_1$ the face that contains $x_1zb$ in its boundary. Let $y_0$ and $y_1$ be such that $bzx_0y_0$ and $bzx_1y_1$ are $4$-paths that are in the boundaries of $f_0$ and $f_1$ respectively. Let us count the charge that $x_0$, $y_0$, and $f_0$ give to $z$ plus half the charge that $f$ gives to $z$. If we show that this sum is at least $\frac{1}{4}$, then by symmetry we will know that $z$ received at least $\frac{1}{2}$ from $x_0$, $x_1$, $y_0$, $y_1$, and the faces $f$, $f_0$, and $f_1$, and that leads to a contradiction.

  Observe that $f_0$ is a $4$-face. If it is a $5^+$-face, then since it has the big vertex $b$ in its boundary, it gives $\frac{1}{4}$ to $z$ in Step $4$, a contradiction.

  Observe that $y_0$ is a small vertex. If $y_0$ is a big vertex, then $y_0$ gives $\frac{1}{4}$ to $z$ in Step $1$, a contradiction. 
  See Figure~\ref{figfinal} for a representation of the vertices we know.

  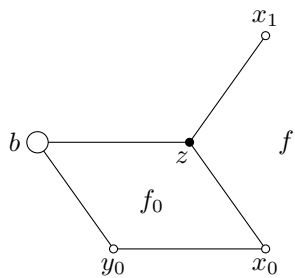
\begin{figure}[h]
    \begin{center}
      \begin{tikzpicture}
        \coordinate (z) at (0,0);
        \coordinate (b) at (-2,0);
        \coordinate (x0) at (1,-1.414);
        \coordinate (x1) at (1,1.414);
        \coordinate (y0) at (-1,-1.414);

        \draw (-0.1,0) node [below] {$z$} ;
        \draw (-2.1,0) node [left] {$b$} ;
        \draw (x0) node [below] {$x_0$} ;
        \draw (x1) node [above] {$x_1$} ;
        \draw (y0) node [below] {$y_0$} ;
        \draw (-0.5,-0.5) node [below] {$f_0$} ;
        \draw (1.5,0) node [left] {$f$} ;

        \draw (b) -- (z);
        \draw (x0) -- (z);
        \draw (x1) -- (z);
        \draw (b) -- (y0);
        \draw (y0) -- (x0);

        \draw [fill=black](z) circle (1.5pt) ;
        \draw [fill=white](b) circle (4pt) ;
        \draw [fill=white](x0) circle (1.5pt) ;
        \draw [fill=white](x1) circle (1.5pt) ;
        \draw [fill=white](y0) circle (1.5pt) ;

      \end{tikzpicture}
    \end{center}
    \caption{The face $f_0$ and the vertex $x_1$.}\label{figfinal}
  \end{figure}

  Observe that $x_0$ has degree $4$. Suppose $x_0$ is a $5^+$-vertex. It gives at least $\frac{1}{4}$ to $z$ in Step $2$, a contradiction. Suppose $x_0$ is a $3$-vertex. Then $x_0$ has a big neighbour by Lemma~\ref{3-b}, and it cannot be $y_0$. This contradicts Lemma~\ref{config1}.

  Let $a$ and $a'$ be the neighbours of $x_0$ distinct from $z$ and $y_0$,
  such that $a$ is consecutive to $z$(as a neighbour of $x_0$). Suppose $a$ is a big vertex. If $x_0$ does not correspond to $v_1$ in Configuration~\ref{config3}, then $x_0$ gives $\frac{1}{4}$ to $z$ in Step $2$. If $x_0$ corresponds to $v_1$ in Configuration~\ref{config3}, then $z$ corresponds to $w_0$ that is not adjacent to two big vertices, so $x_0$ also gives $\frac{1}{4}$ to $z$ in Step $2$. Therefore $a$ is a small vertex.

  Observe that $y_0$ is a $4^+$-vertex. Suppose $y_0$ is a $3$-vertex.  By Lemma~\ref{config2}, there is at least one big vertex in $\{a,a'\}$, which has to be $a'$. If $f$ is a $4$-face, then $x_0$ corresponds to $v_1$ in Configuration~\ref{config3}, and it gives $\frac{1}{4}$ to $z$ in Step $2$. Therefore $f$ is a $5^+$-face, and it gives at least $\frac{1}{5}$ to $z$ in Step $4$, and $x_0$ gives $\frac{1}{5}$ to $z$ in Step $5$. As $\frac{1}{10} + \frac{1}{5} \ge \frac{1}{4}$, this leads to a contradiction.

  \begin{figure}[h!]
    \begin{center}
      \begin{tikzpicture}
        \coordinate (b0) at (0,0);
        \coordinate (w0) at (1,-1);
        \coordinate (w0') at (1.5,-1.5);
        \coordinate (v0) at (1,1);
        \coordinate (v1) at (2,0);
        \coordinate (v2) at (3,1);
        \coordinate (v2') at (4,2);
        \coordinate (v3) at (2,2);
        \coordinate (b1) at (3,-1);
        \coordinate (w1) at (4,0);
        \coordinate (w1') at (5,1);

        \draw (-0.1,0) node [left] {$b_0$} ;
        \draw (w0) node [left] {$w_0$} ;
        \draw (v0) node [left] {$v_0$} ;
        \draw (v1) node [right] {$y_0 = v_1$} ;
        \draw (v2) node [right] {$x_0 = v_2$} ;
        \draw (v3) node [above] {$a' = v_3$} ;
        \draw (3.1,-1) node [right] {$b = b_1$} ;
        \draw (w1) node [right] {$z = w_1$} ;
        \draw (w1') node [right] {$x_1$} ;
        \draw (v2') node [right] {$a$} ;

        \draw (b0) -- (w0);
        \draw (v1) -- (w0);
        \draw (b0) -- (v0);
        \draw (b1) -- (v1);
        \draw (b1) -- (w1);
        \draw (v0) -- (v1);
        \draw (v1) -- (v2);
        \draw (v2) -- (v3);
        \draw (v0) -- (v3);
        \draw (w0) -- (w0');
        \draw (v2) -- (v2');
        \draw (v2) -- (w1);
        \draw (w1) -- (w1');

        \draw [fill=white](b0) circle (4pt) ;
        \draw [fill=black](w0) circle (1.5pt) ;
        \draw [fill=black](v0) circle (1.5pt) ;
        \draw [fill=black](v1) circle (1.5pt) ;
        \draw [fill=black](v2) circle (1.5pt) ;
        \draw [fill=white](v2') circle (1.5pt) ;
        \draw [fill=white](v3) circle (1.5pt) ;
        \draw [fill=white](b1) circle (4pt) ;
        \draw [fill=black](w1) circle (1.5pt) ;
        \draw [fill=white](w1') circle (1.5pt) ;

      \end{tikzpicture}
    \end{center}
    \caption{The case in Lemma~\ref{3vertices} where $y_0$ corresponds to $v_1$ in Configuration~\ref{config3}.\label{fig3vertices}}
  \end{figure}
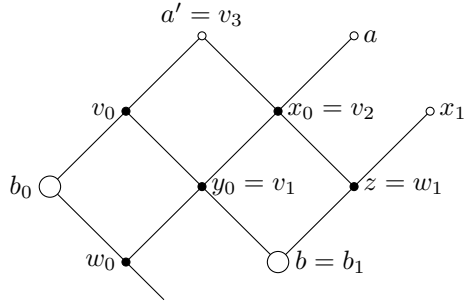

  Suppose first that $y_0$ corresponds to $v_1$ in Configuration~\ref{config3}. See Figure~\ref{fig3vertices} for an illustration of the vertices we know, and of the correspondence with vertices of Configuration~\ref{config3}. By Lemma~\ref{config4}, the third neighbour of $w_0$ is big. Therefore $y_0$ gives $\frac{1}{4}$ to $x_0$ in Step $2$. It follow that $x_0$ gives $\frac{1}{4}$ to $z$ in Step $3$, a contradiction.

  Now $y_0$ does not correspond to $v_1$ in Configuration~\ref{config3}. Vertex $y_0$ gives $\frac{1}{4}$ to $x_0$ in Step $2$, since $x_0$ is a neighbour of $y_0$ consecutive (as a neighbour of $y_0$) to a big neighbour. Therefore $x_0$ gives $\frac{1}{4}$ to $z$ in Step $3$, a contradiction.
\end{proof}

Lemmas~\ref{faces}--\ref{3vertices} conclude the proof of Theorem~\ref{main}.

\section{NP-completeness} \label{complexity}
By Theorem~\ref{main}, there exists a smallest integer $d_0 \le 5$ such that every triangle-free planar graph has an $({\cal F},{\cal F}_{d_0})$-partition. For all $d \ge d_0$, every triangle-free planar graph has an $({\cal F},{\cal F}_{d})$-partition. Let us assume that $d_0 \ge 1$.

In this section, for a fixed $d$ we consider the complexity of the following problem $P_d$: given a triangle-free planar graph $G$, does $G$ have an $({\cal F},{\cal F}_d)$-partition? This can be answered positively in constant time for $d \ge d_0$. However, we prove the following:

\begin{theo}
  For $d < d_0$, the problem $P_d$ is NP-complete.
\end{theo}

The problem is clearly in NP, since checking that a graph is acyclic and/or has degree at most $d$ can be done in polynomial time. Let us show that the problem is NP-hard.

Let $G$ be a counter-example to the property that every triangle-free planar graph admits an $({\cal F},{\cal F}_d)$ partition. We consider such a $G$ with minimum number of vertices, and with minimum number of edges among the counter-examples with minimum number of vertices. Let $e = uv$ be an edge of $G$, and $G' = G - e$. By minimality of $G$, $G'$ admits an $({\cal F},{\cal F}_d)$-partition. In such a partition $(F,D)$, $u$ and $v$ are either both in $F$ or both in $D$, and if they are in $F$, then there is a path from $u$ to $v$ in $G'[F]$ (otherwise it would be an $({\cal F},{\cal F}_d)$-partition of $G$). Observe that in $G'$, $u$ and $v$ are at distance at least $3$, since $G$ is triangle-free. We call a copy of $G'$ an anti-edge $uv$.

We want to make a gadget $H$ with a vertex $x$ that admits an $({\cal F},{\cal F}_d)$-partition, and such that $x$ is in $F$ for all $({\cal F},{\cal F}_d)$-partition $(F,D)$ of $H$.

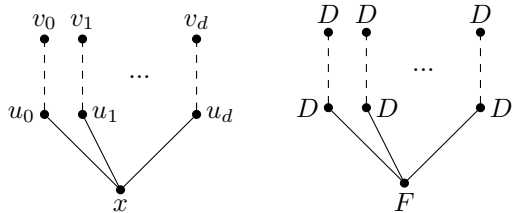
\begin{figure}[h]
  \begin{center}
    \begin{tikzpicture}
      \coordinate (x) at (0,0);
      \coordinate (u0) at (-1,1);
      \coordinate (u1) at (-0.5,1);
      \coordinate (ud) at (1,1);
      \coordinate (v0) at (-1,2);
      \coordinate (v1) at (-0.5,2);
      \coordinate (vd) at (1,2);

      \draw (x) node [below] {$x$} ;
      \draw (u0) node [left] {$u_0$} ;
      \draw (u1) node [right] {$u_1$} ;
      \draw (ud) node [right] {$u_d$} ;
      \draw (v0) node [above] {$v_0$} ;
      \draw (v1) node [above] {$v_1$} ;
      \draw (vd) node [above] {$v_d$} ;
      \draw (0.25,1.5) node {$...$} ;

      \draw (x) -- (u0);
      \draw (x) -- (u1);
      \draw (x) -- (ud);
      \draw [dashed] (v0) -- (u0);
      \draw [dashed] (v1) -- (u1);
      \draw [dashed] (vd) -- (ud);

      \draw [fill=black](x) circle (1.5pt) ;
      \draw [fill=black](u0) circle (1.5pt) ;
      \draw [fill=black](u1) circle (1.5pt) ;
      \draw [fill=black](ud) circle (1.5pt) ;
      \draw [fill=black](v0) circle (1.5pt) ;
      \draw [fill=black](v1) circle (1.5pt) ;
      \draw [fill=black](vd) circle (1.5pt) ;

    \end{tikzpicture}~~~~
    \begin{tikzpicture}
      \coordinate (x) at (0,0);
      \coordinate (u0) at (-1,1);
      \coordinate (u1) at (-0.5,1);
      \coordinate (ud) at (1,1);
      \coordinate (v0) at (-1,2);
      \coordinate (v1) at (-0.5,2);
      \coordinate (vd) at (1,2);

      \draw (x) node [below] {$F$} ;
      \draw (u0) node [left] {$D$} ;
      \draw (u1) node [right] {$D$} ;
      \draw (ud) node [right] {$D$} ;
      \draw (v0) node [above] {$D$} ;
      \draw (v1) node [above] {$D$} ;
      \draw (vd) node [above] {$D$} ;
      \draw (0.25,1.5) node {$...$} ;

      \draw (x) -- (u0);
      \draw (x) -- (u1);
      \draw (x) -- (ud);
      \draw [dashed] (v0) -- (u0);
      \draw [dashed] (v1) -- (u1);
      \draw [dashed] (vd) -- (ud);

      \draw [fill=black](x) circle (1.5pt) ;
      \draw [fill=black](u0) circle (1.5pt) ;
      \draw [fill=black](u1) circle (1.5pt) ;
      \draw [fill=black](ud) circle (1.5pt) ;
      \draw [fill=black](v0) circle (1.5pt) ;
      \draw [fill=black](v1) circle (1.5pt) ;
      \draw [fill=black](vd) circle (1.5pt) ;

    \end{tikzpicture}
  \end{center}
  \caption{The gadget $H$ in Case~\ref{gadgetH1}, and an $({\cal F},{\cal F}_d)$-partition. Dashed lines are anti-edges.}\label{figgadgetH1}
\end{figure}

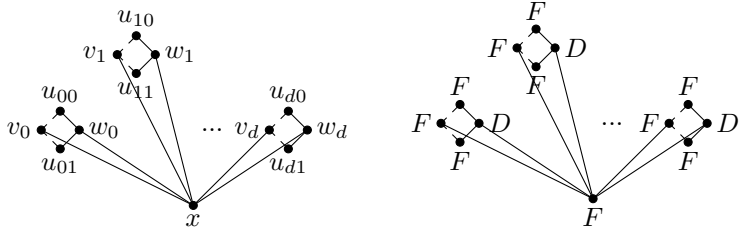
\begin{figure}
  \begin{center}

    \begin{tikzpicture}
      \coordinate (x) at (0,0);
      \coordinate (v0) at (-2,1);
      \coordinate (v1) at (-1,2);
      \coordinate (vd) at (1,1);
      \coordinate (w0) at (-1.5,1);
      \coordinate (w1) at (-0.5,2);
      \coordinate (wd) at (1.5,1);
      \coordinate (u00) at (-1.75,1.25);
      \coordinate (u01) at (-0.75,2.25);
      \coordinate (u0d) at (1.25,1.25);
      \coordinate (u10) at (-1.75,0.75);
      \coordinate (u11) at (-0.75,1.75);
      \coordinate (u1d) at (1.25,0.75);

      \draw (x) node [below] {$x$} ;
      \draw (u00) node [above] {$u_{00}$} ;
      \draw (u01) node [above] {$u_{10}$} ;
      \draw (u0d) node [above] {$u_{d0}$} ;
      \draw (u10) node [below] {$u_{01}$} ;
      \draw (u11) node [below] {$u_{11}$} ;
      \draw (u1d) node [below] {$u_{d1}$} ;
      \draw (v0) node [left] {$v_0$} ;
      \draw (v1) node [left] {$v_1$} ;
      \draw (vd) node [left] {$v_d$} ;
      \draw (w0) node [right] {$w_0$} ;
      \draw (w1) node [right] {$w_1$} ;
      \draw (wd) node [right] {$w_d$} ;
      \draw (0.25,1) node {$...$} ;

      \draw (x) -- (v0);
      \draw (x) -- (v1);
      \draw (x) -- (vd);
      \draw (x) -- (w0);
      \draw (x) -- (w1);
      \draw (x) -- (wd);
      \draw [dashed] (v0) -- (u00);
      \draw [dashed] (v1) -- (u01);
      \draw [dashed] (vd) -- (u0d);
      \draw [dashed] (v0) -- (u10);
      \draw [dashed] (v1) -- (u11);
      \draw [dashed] (vd) -- (u1d);
      \draw (u00) -- (w0);
      \draw (u01) -- (w1);
      \draw (u0d) -- (wd);
      \draw (u10) -- (w0);
      \draw (u11) -- (w1);
      \draw (u1d) -- (wd);

      \draw [fill=black](x) circle (1.5pt) ;
      \draw [fill=black](v0) circle (1.5pt) ;
      \draw [fill=black](v1) circle (1.5pt) ;
      \draw [fill=black](vd) circle (1.5pt) ;
      \draw [fill=black](w0) circle (1.5pt) ;
      \draw [fill=black](w1) circle (1.5pt) ;
      \draw [fill=black](wd) circle (1.5pt) ;
      \draw [fill=black](u00) circle (1.5pt) ;
      \draw [fill=black](u01) circle (1.5pt) ;
      \draw [fill=black](u0d) circle (1.5pt) ;
      \draw [fill=black](u10) circle (1.5pt) ;
      \draw [fill=black](u11) circle (1.5pt) ;
      \draw [fill=black](u1d) circle (1.5pt) ;

    \end{tikzpicture}~~~~
    \begin{tikzpicture}
      \coordinate (x) at (0,0);
      \coordinate (v0) at (-2,1);
      \coordinate (v1) at (-1,2);
      \coordinate (vd) at (1,1);
      \coordinate (w0) at (-1.5,1);
      \coordinate (w1) at (-0.5,2);
      \coordinate (wd) at (1.5,1);
      \coordinate (u00) at (-1.75,1.25);
      \coordinate (u01) at (-0.75,2.25);
      \coordinate (u0d) at (1.25,1.25);
      \coordinate (u10) at (-1.75,0.75);
      \coordinate (u11) at (-0.75,1.75);
      \coordinate (u1d) at (1.25,0.75);

      \draw (x) node [below] {$F$} ;
      \draw (u00) node [above] {$F$} ;
      \draw (u01) node [above] {$F$} ;
      \draw (u0d) node [above] {$F$} ;
      \draw (u10) node [below] {$F$} ;
      \draw (u11) node [below] {$F$} ;
      \draw (u1d) node [below] {$F$} ;
      \draw (v0) node [left] {$F$} ;
      \draw (v1) node [left] {$F$} ;
      \draw (vd) node [left] {$F$} ;
      \draw (w0) node [right] {$D$} ;
      \draw (w1) node [right] {$D$} ;
      \draw (wd) node [right] {$D$} ;
      \draw (0.25,1) node {$...$} ;

      \draw (x) -- (v0);
      \draw (x) -- (v1);
      \draw (x) -- (vd);
      \draw (x) -- (w0);
      \draw (x) -- (w1);
      \draw (x) -- (wd);
      \draw [dashed] (v0) -- (u00);
      \draw [dashed] (v1) -- (u01);
      \draw [dashed] (vd) -- (u0d);
      \draw [dashed] (v0) -- (u10);
      \draw [dashed] (v1) -- (u11);
      \draw [dashed] (vd) -- (u1d);
      \draw (u00) -- (w0);
      \draw (u01) -- (w1);
      \draw (u0d) -- (wd);
      \draw (u10) -- (w0);
      \draw (u11) -- (w1);
      \draw (u1d) -- (wd);

      \draw [fill=black](x) circle (1.5pt) ;
      \draw [fill=black](v0) circle (1.5pt) ;
      \draw [fill=black](v1) circle (1.5pt) ;
      \draw [fill=black](vd) circle (1.5pt) ;
      \draw [fill=black](w0) circle (1.5pt) ;
      \draw [fill=black](w1) circle (1.5pt) ;
      \draw [fill=black](wd) circle (1.5pt) ;
      \draw [fill=black](u00) circle (1.5pt) ;
      \draw [fill=black](u01) circle (1.5pt) ;
      \draw [fill=black](u0d) circle (1.5pt) ;
      \draw [fill=black](u10) circle (1.5pt) ;
      \draw [fill=black](u11) circle (1.5pt) ;
      \draw [fill=black](u1d) circle (1.5pt) ;

    \end{tikzpicture}
  \end{center}
  \caption{The gadget $H$ in Case~\ref{gadgetH2}, and an $({\cal F},{\cal F}_d)$-partition.}\label{figgadgetH2}
\end{figure}

We construct $H$ as follows:

\begin{enumerate}
\item \label{gadgetH1}
  Suppose for all $({\cal F},{\cal F}_d)$-partition $(F,D)$ of $G'$, $u$ and $v$ are in $D$. See Figure~\ref{figgadgetH1} for an illustration of the construction of $H$ and an $({\cal F},{\cal F}_d)$-partition of $H$ in this case. Take $d+1$ copies of $G'$, called $G'_0$, ..., $G'_d$, and add a new vertex $x$ adjacent to each copy of $u$. Consider an $({\cal F},{\cal F}_d)$-partition $(F,D)$ of $G'$. This leads to an $({\cal F},{\cal F}_d)$-partition $(F_i,D_i)$ of each $G_i$, and $(\bigcup_iF_i \cup \{x\}, \bigcup_iD_i)$ is an $({\cal F},{\cal F}_d)$-partition of $H$. 

  Let us now prove that for any $({\cal F},{\cal F}_d)$-partition $(F,D)$ of $H$, $x$ belongs to $F$. For any $({\cal F},{\cal F}_d)$-partition $(F,D)$ of $H$, if $x \in D$, then there exists a $u_i$ that is in $F$, so the corresponding $G'_i$ admits an $({\cal F},{\cal F}_d)$-partition with $u_i \in F$, a contradiction.

\item \label{gadgetH2}
  Suppose there exists an $({\cal F},{\cal F}_d)$-partition $(F,D)$ of $G'$ such that $u$ and $v$ are in $F$. See Figure~\ref{figgadgetH2} for an illustration of the construction of $H$ and an $({\cal F},{\cal F}_d)$-partition of $H$ in this case. We construct $H$ as follows.
  Consider a vertex $x$. We add new vertices $v_0,...,v_d$ and $w_0,...,w_d$ to the graph, adjacent to $x$. Then for $0 \le i \le d$ and $0 \le j \le 1$, we add a new vertex $u_{ij}$, the anti-edge $v_iu_{ij}$, and the edge $u_{ij}w_i$.

  Graph $H$ admits an $({\cal F},{\cal F}_d)$-partition. Indeed, consider an $({\cal F},{\cal F}_d)$-partition of $G'$ with $u$ and $v$ in $F$, and apply it to every anti-edge of $H$ (as before, we take the union of the $F_i$ and the union of the $D_i$). Then the $v_i$ and $u_{ij}$ are all in $F$. Add all the $w_i$ to $D$. Add $x$ to $F$. We then have an $({\cal F},{\cal F}_d)$-partition of $H$.

  Let us now prove that for any $({\cal F},{\cal F}_d)$-partition $(F,D)$ of $H$, $x$ belongs to $F$. For any $({\cal F},{\cal F}_d)$-partition $(F,D)$ of $H$, if $x \in D$, then there exists an $i$ such that $v_i$ and $w_i$ are in $F$, thus $u_{i0}$ and $u_{i1}$ are in $F$, so there is a cycle in $H[F]$, a contradiction.
\end{enumerate}

\begin{figure}[h]
  \begin{center}
    \begin{tikzpicture}
      \coordinate (y) at (0,0);
      \coordinate (x1) at (-1,1);
      \coordinate (x2) at (0,2);
      \coordinate (x3) at (1,1);

      \draw (y) node [below] {$y$} ;
      \draw (y) node [left] {$D$} ;
      \draw (x1) node [above] {$F$} ;
      \draw (x2) node [left] {$F$} ;
      \draw (x3) node [above] {$F$} ;
      \draw (-2,1) node [left] {$H$} ;
      \draw (0,3) node [above] {$H$} ;
      \draw (2,1) node [right] {$H$} ;

      \draw (y) -- (x1);
      \draw (y) -- (x3);
      \draw (x2) -- (x1);
      \draw (x3) -- (x2);
      \draw [-,>=latex] (x1) to[out = 190, in = -90] ++(-1,0) to[out = 90, in = 170] ++(1,0);
      \draw [-,>=latex] (x2) to[out = 100, in = 180] ++(0,1) to[out = 0, in = 80] ++(0,-1);
      \draw [-,>=latex] (x3) to[out = 10, in = 90] ++(1,0) to[out = -90, in = -10] ++(-1,0);

      \draw [fill=black](y) circle (1.5pt) ;
      \draw [fill=black](x1) circle (1.5pt) ;
      \draw [fill=black](x2) circle (1.5pt) ;
      \draw [fill=black](x3) circle (1.5pt) ;

    \end{tikzpicture}
  \end{center}
  \caption{The gadget $H'$ with an $({\cal F},{\cal F}_d)$-partition.}\label{figgadgetH'}
\end{figure}
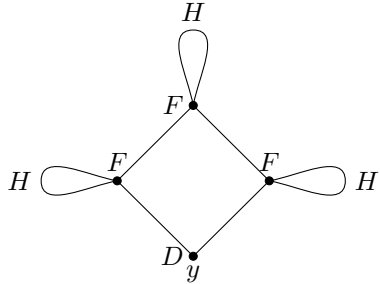

Observe that we can make a gadget $H'$ with a vertex $y$ that admits an $({\cal F},{\cal F}_d)$-partition, and such that $y$ is in $D$ for all $({\cal F},{\cal F}_d)$-partition $(F,D)$ of $H'$ (see Figure~\ref{figgadgetH'}): we take three copies of $H$, and make a $4$-cycle with the corresponding copies of $x$ and a new vertex $y$. Taking an $({\cal F},{\cal F}_d)$-partition of $H$ for each copy of $H$, and adding $y$ to $D$ leads to an $({\cal F},{\cal F}_d)$-partition of $H'$. Conversely, in an $({\cal F},{\cal F}_d)$-partition $(F,D)$ of $H'$, all the copies of $x$ are in $F$, so $y$ is in $D$.

We will first make a reduction from the problem {\sc Planar $3$-sat} to $P_0$, and then from $P_0$ to $P_d$ with $d < d_0$.

\subsection*{First reduction: from {\sc Planar $3$-sat} to $P_0$}
Here we will use the gadget $H$ for $d = 0$.

Consider an instance $I$ of {\sc Planar $3$-sat}. The instance $I$ is a boolean formula in conjunctive normal form, associated to a planar graph $G_I$. For each clause $C$ of $I$ with variables $x$, $y$ and $z$, we make a $4$-cycle $K_C = x_Cy_Cz_Ca_C$. For each variable $x$ that appears $k_x$ times in the formula, we make the following gadget $G_x$ a path $p_{x,0}...p_{x,2k_x-1}$, and for all $i \in [0,2k_x-2]$ we add two adjacent vertices, $q_{x,i}$ and $r_{x,i+1}$, adjacent to $p_{x,i}$ and $p_{x,i+1}$ respectively (see Figure~\ref{figGX}). We then add a copy of $H$ for each clause $C$ such that $a_C$ corresponds to the vertex $x$ of $H$, and a copy of $H$ for each $q_{x,i}$ and each $r_{x,i}$ such that $q_{x,i}$ and $r_{x,i}$ respectively correspond to the vertex $x$ of $H$. Then for every clause $C$ and every variable $x$ that appears in $C$, we add an edge from $x_C$ to a $p_{x,i}$, with an even $i$ if the literal associated to $x$ in $C$ is a positive literal and an odd $i$ otherwise, such that no two $x_C$ are adjacent to the same $p_{x,i}$ (see Figure~\ref{figKC}). It is possible to do so without breaking planarity, since the graph $G_I$ is planar. We call $G'_I$ the graph we obtain.

\begin{figure}[h]
  \begin{center}
    \begin{tikzpicture}
      \coordinate (a) at (0,0);
      \coordinate (x) at (-1,1);
      \coordinate (y) at (0,2);
      \coordinate (z) at (1,1);

      \node[draw,ellipse] (X) at (-2,3) {$G_x$};
      \node[draw,ellipse] (Y) at (0,3) {$G_y$};
      \node[draw,ellipse] (Z) at (2,3) {$G_z$};

      \draw (x) node [below left] {$x_C$} ;
      \draw (y) node [left] {$y_C$} ;
      \draw (z) node [below right] {$z_C$} ;
      \draw (a) node [left] {$a_C$} ;
      \draw (-0.2,-0.5) node [left] {$H$} ;

      \draw (a) -- (x);
      \draw (a) -- (z);
      \draw (y) -- (x);
      \draw (z) -- (y);
      \draw [-,>=latex] (a) to[out = 290, in = -10] ++(0,-1) to[out = 180, in = 250] ++(0,1);
      \draw (x) -- (X);
      \draw (y) -- (Y);
      \draw (z) -- (Z);

      \draw [fill=black](a) circle (1.5pt) ;
      \draw [fill=black](x) circle (1.5pt) ;
      \draw [fill=black](y) circle (1.5pt) ;
      \draw [fill=black](z) circle (1.5pt) ;

    \end{tikzpicture}~~~~
    \begin{tikzpicture}
      \coordinate (a) at (0,0);
      \coordinate (x) at (-1,1);
      \coordinate (y) at (0,2);
      \coordinate (z) at (1,1);

      \node[draw,ellipse] (X) at (-2,3) {$G_x$};
      \node[draw,ellipse] (Y) at (0,3) {$G_y$};
      \node[draw,ellipse] (Z) at (2,3) {$G_z$};

      \draw (x) node [below left] {$D$} ;
      \draw (y) node [left] {$F$} ;
      \draw (z) node [below right] {$F$} ;
      \draw (a) node [left] {$F$} ;
      \draw (-0.2,-0.5) node [left] {$H$} ;

      \draw (a) -- (x);
      \draw (a) -- (z);
      \draw (y) -- (x);
      \draw (z) -- (y);
      \draw [-,>=latex] (a) to[out = 290, in = -10] ++(0,-1) to[out = 180, in = 250] ++(0,1);
      \draw (x) -- (X);
      \draw (y) -- (Y);
      \draw (z) -- (Z);

      \draw [fill=black](a) circle (1.5pt) ;
      \draw [fill=black](x) circle (1.5pt) ;
      \draw [fill=black](y) circle (1.5pt) ;
      \draw [fill=black](z) circle (1.5pt) ;

    \end{tikzpicture}
  \end{center}
  \caption{The cycle $K_C$ of a clause $C$ with variables $x$, $y$ and $z$, and an $({\cal F},{\cal F}_d)$-partition in the case where variable $x$ satisfies the clause.}\label{figKC}
\end{figure}
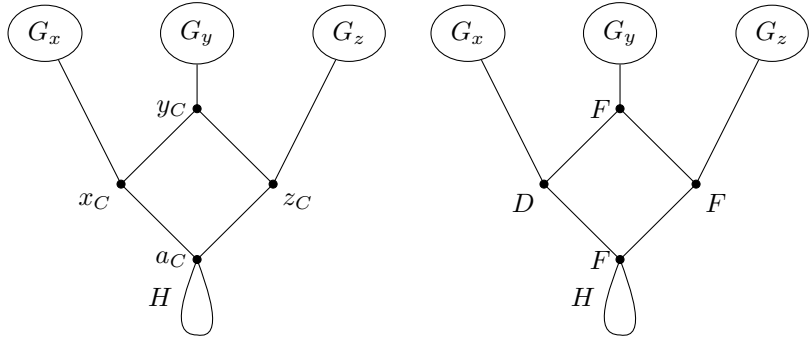

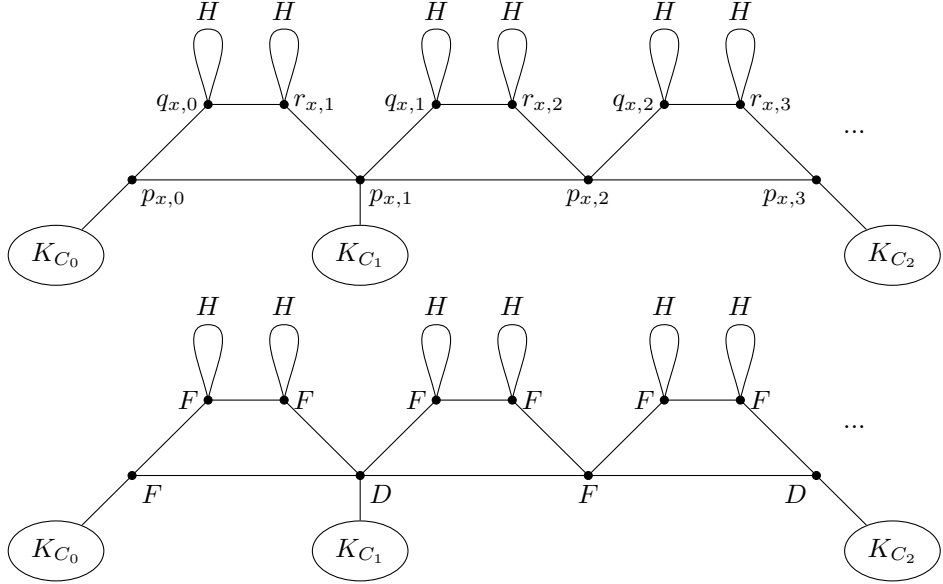
\begin{figure}[h!]
  \begin{center}
    \begin{tikzpicture}
      \coordinate (p0) at (0,0);
      \coordinate (p1) at (3,0);
      \coordinate (p2) at (6,0);
      \coordinate (p3) at (9,0);
      \coordinate (q0) at (1,1);
      \coordinate (r1) at (2,1);
      \coordinate (q1) at (4,1);
      \coordinate (r2) at (5,1);
      \coordinate (q2) at (7,1);
      \coordinate (r3) at (8,1);

      \node[draw,ellipse] (C0) at (-1,-1) {$K_{C_0}$};
      \node[draw,ellipse] (C1) at (3,-1) {$K_{C_1}$};
      \node[draw,ellipse] (C2) at (10,-1) {$K_{C_2}$};

      \draw (p0) node [below right] {$p_{x,0}$} ;
      \draw (p1) node [below right] {$p_{x,1}$} ;
      \draw (p2) node [below] {$p_{x,2}$} ;
      \draw (p3) node [below left] {$p_{x,3}$} ;

      \draw (q0) node [left] {$q_{x,0}$} ;
      \draw (q1) node [left] {$q_{x,1}$} ;
      \draw (q2) node [left] {$q_{x,2}$} ;
      \draw (r1) node [right] {$r_{x,1}$} ;
      \draw (r2) node [right] {$r_{x,2}$} ;
      \draw (r3) node [right] {$r_{x,3}$} ;
      \draw (1,2) node [above] {$H$} ;
      \draw (2,2) node [above] {$H$} ;
      \draw (4,2) node [above] {$H$} ;
      \draw (5,2) node [above] {$H$} ;
      \draw (7,2) node [above] {$H$} ;
      \draw (8,2) node [above] {$H$} ;
      \draw (9.5,0.5) node [above] {$...$} ;

      \draw (p0) -- (p1);
      \draw (p1) -- (p2);
      \draw (p2) -- (p3);
      \draw (p0) -- (q0);
      \draw (p1) -- (q1);
      \draw (p2) -- (q2);
      \draw (p1) -- (r1);
      \draw (p2) -- (r2);
      \draw (p3) -- (r3);
      \draw (q0) -- (r1);
      \draw (q1) -- (r2);
      \draw (q2) -- (r3);
      \draw [-,>=latex] (q0) to[out = 100, in = 180] ++(0,1) to[out = 0, in = 80] ++(0,-1);
      \draw [-,>=latex] (q1) to[out = 100, in = 180] ++(0,1) to[out = 0, in = 80] ++(0,-1);
      \draw [-,>=latex] (q2) to[out = 100, in = 180] ++(0,1) to[out = 0, in = 80] ++(0,-1);
      \draw [-,>=latex] (r1) to[out = 100, in = 180] ++(0,1) to[out = 0, in = 80] ++(0,-1);
      \draw [-,>=latex] (r2) to[out = 100, in = 180] ++(0,1) to[out = 0, in = 80] ++(0,-1);
      \draw [-,>=latex] (r3) to[out = 100, in = 180] ++(0,1) to[out = 0, in = 80] ++(0,-1);

      \draw (p0) -- (C0);
      \draw (p1) -- (C1);
      \draw (p3) -- (C2);

      \draw [fill=black](p0) circle (1.5pt) ;
      \draw [fill=black](p1) circle (1.5pt) ;
      \draw [fill=black](p2) circle (1.5pt) ;
      \draw [fill=black](p3) circle (1.5pt) ;
      \draw [fill=black](q0) circle (1.5pt) ;
      \draw [fill=black](q1) circle (1.5pt) ;
      \draw [fill=black](q2) circle (1.5pt) ;
      \draw [fill=black](r1) circle (1.5pt) ;
      \draw [fill=black](r2) circle (1.5pt) ;
      \draw [fill=black](r3) circle (1.5pt) ;

    \end{tikzpicture}

    \begin{tikzpicture}
      \coordinate (p0) at (0,0);
      \coordinate (p1) at (3,0);
      \coordinate (p2) at (6,0);
      \coordinate (p3) at (9,0);
      \coordinate (q0) at (1,1);
      \coordinate (r1) at (2,1);
      \coordinate (q1) at (4,1);
      \coordinate (r2) at (5,1);
      \coordinate (q2) at (7,1);
      \coordinate (r3) at (8,1);

      \node[draw,ellipse] (C0) at (-1,-1) {$K_{C_0}$};
      \node[draw,ellipse] (C1) at (3,-1) {$K_{C_1}$};
      \node[draw,ellipse] (C2) at (10,-1) {$K_{C_2}$};

      \draw (p0) node [below right] {$F$} ;
      \draw (p1) node [below right] {$D$} ;
      \draw (p2) node [below] {$F$} ;
      \draw (p3) node [below left] {$D$} ;

      \draw (q0) node [left] {$F$} ;
      \draw (q1) node [left] {$F$} ;
      \draw (q2) node [left] {$F$} ;
      \draw (r1) node [right] {$F$} ;
      \draw (r2) node [right] {$F$} ;
      \draw (r3) node [right] {$F$} ;
      \draw (1,2) node [above] {$H$} ;
      \draw (2,2) node [above] {$H$} ;
      \draw (4,2) node [above] {$H$} ;
      \draw (5,2) node [above] {$H$} ;
      \draw (7,2) node [above] {$H$} ;
      \draw (8,2) node [above] {$H$} ;
      \draw (9.5,0.5) node [above] {$...$} ;

      \draw (p0) -- (p1);
      \draw (p1) -- (p2);
      \draw (p2) -- (p3);
      \draw (p0) -- (q0);
      \draw (p1) -- (q1);
      \draw (p2) -- (q2);
      \draw (p1) -- (r1);
      \draw (p2) -- (r2);
      \draw (p3) -- (r3);
      \draw (q0) -- (r1);
      \draw (q1) -- (r2);
      \draw (q2) -- (r3);
      \draw [-,>=latex] (q0) to[out = 100, in = 180] ++(0,1) to[out = 0, in = 80] ++(0,-1);
      \draw [-,>=latex] (q1) to[out = 100, in = 180] ++(0,1) to[out = 0, in = 80] ++(0,-1);
      \draw [-,>=latex] (q2) to[out = 100, in = 180] ++(0,1) to[out = 0, in = 80] ++(0,-1);
      \draw [-,>=latex] (r1) to[out = 100, in = 180] ++(0,1) to[out = 0, in = 80] ++(0,-1);
      \draw [-,>=latex] (r2) to[out = 100, in = 180] ++(0,1) to[out = 0, in = 80] ++(0,-1);
      \draw [-,>=latex] (r3) to[out = 100, in = 180] ++(0,1) to[out = 0, in = 80] ++(0,-1);

      \draw (p0) -- (C0);
      \draw (p1) -- (C1);
      \draw (p3) -- (C2);

      \draw [fill=black](p0) circle (1.5pt) ;
      \draw [fill=black](p1) circle (1.5pt) ;
      \draw [fill=black](p2) circle (1.5pt) ;
      \draw [fill=black](p3) circle (1.5pt) ;
      \draw [fill=black](q0) circle (1.5pt) ;
      \draw [fill=black](q1) circle (1.5pt) ;
      \draw [fill=black](q2) circle (1.5pt) ;
      \draw [fill=black](r1) circle (1.5pt) ;
      \draw [fill=black](r2) circle (1.5pt) ;
      \draw [fill=black](r3) circle (1.5pt) ;

    \end{tikzpicture}
  \end{center}
  \caption{The gadget $G_x$ for a variable $x$, with an $({\cal F},{\cal F}_d)$-partition that corresponds to the assignation of $x$ to true. Here the literal associated to $x$ in $C_0$ is positive, and that associated to $x$ in $C_1$ and $C_2$ is negative.}\label{figGX}
\end{figure}

Suppose $I$ is satisfiable, and let us consider an assignation $\sigma$ of the variables that satisfies $I$. Let us make an $({\cal F},{\cal F}_0)$-partition of $G'_I$. We first take an $({\cal F},{\cal F}_0)$-partition for each copy of $H$. All the $a_C$, $q_{x,i}$ and $r_{x,i}$ are in $F$.
For each variable $x$, if $\sigma(x) = 1$, then we put all the $p_{x,2i}$ in $F$ and the $p_{x,2i+1}$ in $D$, else we put all the $p_{x,2i}$ in $D$ and the $p_{x,2i+1}$ in $F$. Then for each clause $C$, we choose a variable $x$ of $C$ that satisfies the clause (i.e. $x$ is true if the literal associated to $x$ in $C$ is a positive literal, and false otherwise), we put $x_C$ in $D$ and for the two other variables of $C$, we put the corresponding $y_C$ in $F$. 

All the vertices are in $F$ or in $D$. Let $v$ be a $G'_I$-vertex in $D$. If $v$ is in a copy of $H$, then it has no neighbour in $D$. If $v$ is a $x_C$, then the three other vertices of $K_C$ are in $F$. If $v$ is a $p_{x,i}$, then $p_{x,i+1}$ and $p_{x,i-1}$ are in $F$ if they exist, and all the $q_j$ and $r_j$ are in $F$. Suppose there are two $G'_I[F]$-neighbours in $D$. One is a $x_C$ and the other is a $p_{x,i}$ (with the same $x$). Then by construction the variable $x$ satisfies clause $C$ (i.e. $x$ is true if the literal associated to $x$ in $C$ is a positive literal, and false otherwise). If $x$ is associated to a positive literal in clause $C$, then $\sigma(x) = 1$ and $i$ is even, thus $p_{x,i}$ is in $F$, a contradiction. If $x$ is associated to a negative literal in clause $C$, then $\sigma(x) = 0$ and $i$ is odd, thus $p_{x,i}$ is in $F$, a contradiction. Graph $G'_I[F]$ has no cycle: there is no cycle in the copies of $H$ with every vertex in $F$; for each clause $C$, $K_C$ has a vertex in $D$, and for each $i \in [0,2k_x-2]$, $p_{x,2i}$ or $p_{x,2i+1}$ is in $D$. Therefore $(F,D)$ is an $({\cal F},{\cal F}_0)$-partition of $G'_I$.

Suppose now that there is an $({\cal F},{\cal F}_0)$-partition $(F,D)$ of $G'_I$. All the $a_C$, the $q_{x,i}$ and the $r_{x,i}$ are in $F$. For all variable $x$ and all $i \in [0,2k_x-2]$, either $p_{x,i} \in F$ and $p_{x,i+1} \in D$, or $p_{x,i} \in D$ and $p_{x,i+1} \in F$. Therefore for all $x$, either all the $p_{x,i}$ are in $F$ for $i$ even and in $D$ for $i$ odd, or all the $p_{x,i}$ are in $D$ for $i$ even and in $F$ for $i$ odd. Let $\sigma$ be the assignation of the variables $x$ such that $\sigma(x) = 1$ if $p_{x,0}$ is in $F$, and $\sigma(x) = 0$ otherwise. Let $C$ be a clause of $I$. At least one of the $x_C$ is in $D$ (otherwise $K_C$ is a cycle with every vertex in $F$), and it is adjacent to a $p_{x,i}$ with $i$ even if $x$ is positive and $i$ odd if $x$ is negative in $C$. This $p_{x,i}$ is in $F$, so if $x$ is positive in $C$, then $\sigma(x) = 1$, else $\sigma(x) = 0$. Therefore $\sigma$ satisfies clause $C$, and this is true for all $C$, so $\sigma$ satisfies $I$.

It is easy to see that the reduction is polynomial, and that $G'_I$ is a triangle-free planar graph. Thus this is a polynomial reduction from {\sc Planar $3$-sat} to $P_0$.

\subsection*{Second reduction: from $P_0$ to $P_d$ with $d < d_0$}
Consider an instance $I$ of $P_0$. For each vertex $v$ in $I$, add $d$ copies of $H'$, such that the corresponding copies of $y$ are adjacent to $v$. We call $I_d$ the resulting graph.

Suppose $I$ admits an $({\cal F},{\cal F}_0)$-partition. Consider an $({\cal F},{\cal F}_d)$-partition of $H'$. Apply it to every copy of $H'$ we made in $I_d$. Complete it with an $({\cal F},{\cal F}_0)$-partition of $I$. The obtained partition is an $({\cal F},{\cal F}_d)$-partition of $I_d$.

Suppose now that $I_d$ admits an $({\cal F},{\cal F}_d)$-partition $(F,D)$. In each copy of $H'$, we have $y \in D$, so each vertex in $I$ has exactly $d$ $(I_d-V(I))$-neighbours in $D$ and no $(I_d-V(I))$-neighbours in $F$. Therefore $(F \cap V(I), D \cap V(I))$ is an $({\cal F},{\cal F}_0)$-partition of $I$.

It is easy to see that the reduction is polynomial, and that $I_d$ is a triangle-free planar graph. Thus this is a polynomial reduction from $P_0$ to $P_d$.

\section{Acknowledgements}

We are deeply grateful to Pascal Ochem who pointed out that the PLANAR 3-SAT problem would be helpful to prove our complexity result.

Moreover, this research was partially supported by ANR EGOS project, under contract ANR-12-JS02-002-01.

\bibliographystyle{plain}
\bibliography{biblio} {}

\end{document}